%% file: ms.tex
\documentclass[11pt]{article}

\usepackage[margin=1in]{geometry}

\usepackage{soul}
\usepackage{url}
\usepackage[hidelinks]{hyperref}
\usepackage[utf8]{inputenc}
\usepackage[small]{caption}
\usepackage{graphicx}
\usepackage{subcaption}
\usepackage{amsmath}
\usepackage{booktabs}
\usepackage{algorithm}
\usepackage{algorithmic}
\usepackage{amsfonts}
\usepackage{amsthm}
\usepackage{hyperref}
\usepackage{bm}
\newtheorem{example}{Example}[section]
\newtheorem{definition}{Definition}[section]
\newtheorem{theorem}{Theorem}[section]

\newtheorem{lemma}[theorem]{Lemma}
\newtheorem{claim}[theorem]{Claim}
\newtheorem{observation}[theorem]{Observation}

\newcommand{\E}{\mathbb{E}}
\newcommand{\RR}{\mathbb{R}}      
\newcommand{\NN}{\mathbb{N}}      

\newcommand{\allocation}{\bm{X}}
\newcommand{\payment}{\bm{p}}
\newcommand{\outcome}{\langle\allocation, \payment \rangle}

\newcommand{\ICRegret}{\text{IC-Regret}}
\newcommand{\ICEnvy}{\text{IC-Envy}}
\newcommand{\gsp}{\text{gsp}}

\usepackage[usenames,dvipsnames,svgnames,table]{xcolor}

\title{Envy, Regret, and Social Welfare Loss}

\author{Riccardo Colini-Baldeschi\thanks{Facebook, Core Data Science, \texttt{rickuz@fb.com}} \and Stefano Leonardi\thanks{Sapienza University of Rome, Department of Computer and Systems Sciences and Facebook, Core Data Science, \texttt{leonardi@diag.uniroma1.it}}\thanks{Partially supported by ERC Advanced Grant  788893 AMDROMA "Algorithmic and Mechanism Design Research in Online Markets" and MIUR PRIN project ALGADIMAR "Algorithms, Games, and Digital Markets"} \and Okke Schrijvers\thanks{Facebook, Core Data Science, \texttt{okke@fb.com}} \and Eric Sodomka\thanks{Facebook, Core Data Science, \texttt{sodomka@fb.com}}}
\date{}

\begin{document}

\maketitle

\begin{abstract}

Incentive compatibility (IC) is one of the most fundamental properties of an auction mechanism, including those used for online advertising. Recent methods \cite{FSS19,LMSV18} show that counterfactual runs of the auction mechanism with different bids can be used to determine whether an auction is IC. In this paper we show that a similar result can be obtained by looking at the advertisers' \textit{envy}, which can be computed with one single execution of the auction. We introduce two metrics to evaluate the incentive-compatibility of an auction: \ICRegret{} and \ICEnvy{}. For position auction environments, we show that for a large class of pricing schemes (which includes e.g. VCG and GSP), $\ICEnvy{} \ge \ICRegret{}$ (and $\ICEnvy{} = \ICRegret{}$ when bids are distinct). We consider non-separable discounts in the Ad Types environment \cite{CMSW19} where we show that for a generalization of GSP also $\ICEnvy{} \ge \ICRegret{}$. Our final theoretical result is that in all these settings $\ICEnvy{}$ be used to bound the loss in social welfare due advertiser misreports.

Finally, we show that \ICEnvy{} is useful as a feature to predict \ICRegret{} in auction environments beyond the ones for which we show theoretical results. In particular, using \ICEnvy{}  yields better results than training models using only price and value features.

%
%
\end{abstract}

\maketitle

\input{introduction}
\input{preliminaries}

\input{truthful}
\input{extendedGSP-monotonic}

\input{swl-Matching}
\input{ml-experiments}

\input{conclusion}


\bibliography{references}
\bibliographystyle{plain}

\appendix
\input{appendix}

\end{document}

%% file: introduction.tex
\section{Introduction}
\label{introduction}

Over the past decades, online advertising has grown into a huge industry, with many different online publishers offering impression opportunities. 
Auction theory has played a major role in shaping this ecosystem, and many ad auctions strive to be \emph{Incentive Compatible} (IC), which means that an advertiser achieves the best outcome by truthfully reporting their willingness-to-pay. Despite the role that auction theory has played, the resulting systems may not be IC. For example, intermediaries (called Demand Side Platforms or DSPs) first run an auction to determine the best ad among their clients, and then pass this along to a publisher who runs their own auction including bids from other sources. Even when both auctions in isolation are IC, their composition is not. Furthermore, some publishers use past bids to set a reserve price (or minimum bid), and others are moving to a pay-your-bid model entirely (which has strong incentives to misreport willingness-to-pay) \cite{adexchanger2017firstprice, adexchanger2019googlefirstprice}. 

Recently, there have been several works addressing the problem of determining whether an auction is IC based on statistical tests using only inputs and outputs of an (unobserved) auction mechanism \cite{LMSV18,FSS19}. This gives advertisers the power to test whether an auction is IC without having access to the code. Feng et al. \cite{FSS19} proposed to use regret \cite{FS97} as a way to measure ``how far'' an auction is from being IC:
\begin{align}
\ICRegret_i(v_i) = \max_{b_i} \E_{b_{-i}}\left[ u_i(b_i, b_{-i}) - u_i(v_i, b_{-i})\right],\label{eq:ic-regret}
\end{align}
where $v_i$ is the true value of advertiser $i$, $b_i$ the bid of $i$, $b_{-i}$ the bids of other advertisers, and $u_i(\cdot)$ the (expected) utility of $i$. \ICRegret{} captures the difference in utility between bidding truthfully, and the maximum utility achievable. By definition, IC mechanisms have \ICRegret{} $0$, while higher \ICRegret{} indicates a stronger incentive to misreport.

While measuring IC is most naturally a concern for advertisers (who cannot observe the auction mechanism directly), it is also both important and non-trivial for the auctioneer. The auctioneer cares about IC auctions because they admit simple optimal bidding strategies (namely truthful reporting), and simple bidding strategies may in turn lead to lower churn of advertisers.\footnote{Moreover, additional advertisers are better for revenue than being clever about devising a revenue-optimal auction mechanism, see \cite{bulow1994auctions} (and follow-ups).} Additionally, when auctions are not IC, bidders don't truthfully report their value, which may harm social welfare\footnote{While in some cases there are symmetric equilibria in which social welfare is not harmed \cite{V07}, these need not always exist and bidders may not reach equilibrium \cite{daskalakis2009complexity,babichenko2017communication,rubinstein2018inapproximability}.} and thus the quality of the service provided to the advertisers. While important, it may not be straightforward for publishers to guarantee that their auction is IC for many different reasons: complex interaction between different layers in the advertising ecosystem, running-time constraints, bugs in the auction code, and so forth. 

Feng et al. \cite{FSS19} proposed a method to determine \ICRegret{} for publishers (by taking a worst-case perspective over the advertiser value $v_i$). A downside of their method is that it requires many counterfactual evaluations of the auction's outcomes for alternative bids. This means that the auction code needs to be run many times over. While this may be the best thing one can do with only black-box access to the auction mechanism, could we do better if we're using intermediate data from the auction mechanism?

To overcome the practical difficulties to measure \ICRegret{}, we propose to use Envy \cite{foley1967resource} as a proxy for \ICRegret{}, by identifying relevant classes of auction mechanisms where Envy and \ICRegret{} coincide or where \ICRegret{} is upper bounded by Envy. So what is Envy? Instead of comparing the advertiser's utility against her utility for alternative bids, Envy takes a single outcome, and measures to what extend advertisers are happy with the outcome. Ad auctions typically simultaneously sell multiple ad slots with varying click-through rates \cite{EOS07,V07}. Let ${\bf x}$ be an expected allocation vector (where allocation corresponds to the ad being clicked) and $\bf p$ be an expected pricing vector.	Envy of bidder $i$ given the outcome $({\bf x}, {\bf p})$ is:

\begin{align}
	\text{Envy}_i(v_i, {\bf x}, {\bf p}) = \max_j \left(x_j \cdot  v_i - p_j\right) - \left(x_i\cdot  v_i - p_i\right)
\end{align}

Envy is defined with respect to some outcome. However, in an auction, changing ones bid may change the outcome in the auction. Therefore, we define \ICEnvy{} as the Envy experienced in the outcome \emph{when a bidder bids truthfully}. In the following, let ${\bf x}(v_i, b_{-i})$ is the expected allocation of all bidders and ${\bf p}(v_i, b_{-i})$ be the expected payment vector of all bidders, when bidder $i$ bids truthfully and for bids $b_{-i}$ of the remaining bidders. \ICEnvy{} is then:
	\begin{align}
	\ICEnvy(v_i) &= \max_j \left(x_j(v_i, b_{-i}) \cdot  v_i - p_j(v_i, b_{-i})\right) \\ 
				      & \qquad-
				       \left(x_i(v_i, b_{-i}) \cdot  v_i - p_i(v_i, b_{-i})\right)
	\end{align}
In arbitrary auction environments, \ICEnvy{} and \ICRegret{} don't necessarily coincide. There are natural auction environments with envy-free outcomes, that still have positive \ICRegret{} and vice versa.
\begin{example}[$\ICEnvy = 0$, \ICRegret{} is positive]\label{ex:fp}
	Consider a single-item, first-price auction, with two bidders with values $v_1 = \$10$ and $v_2 = \$8$ and assume bidder 1 bids truthfully. The \ICRegret{} for bidder 1 is $\$2-\epsilon$ for arbitrary small $\epsilon$, as the best alternative bid for them is $\$8 + \epsilon$. However, \ICEnvy{} is $\$ 0$ as the only alternative allocation for bidder 1 is to not receive the item.
\end{example}
\begin{example}[$\ICRegret = 0$, \ICEnvy{} is positive]
	Consider the IC auction for a single item with 2 bidders who face different reserve prices\footnote{While non-anonymous reserve prices may seem contrived, they naturally occur e.g. for the revenue-optimal auction on non-i.i.d. bidders \cite{M81}.} $r_1=\$1$ and $r_2=\$5$. With bids $b_1 = b_2 = \$3$, bidder 1 received the good at their reserve price of $\$1$. Bidder 2 has \ICEnvy{} of $\$2$, but there is no counterfactual bid that will give her positive utility (hence $\ICRegret = 0$).
\end{example}

While in general \ICRegret{} and \ICEnvy{} can be quite different, in this paper we show that there are large auction classes for which $\ICRegret \le \ICEnvy$ (and under mild conditions $\ICRegret = \ICEnvy$). Since computing $\ICEnvy$ requires no counterfactual evaluation of the algorithm, it can serve as an efficient certificate that \ICRegret{} is low.\footnote{For an auctioneer who cares about incentive compatibility, false positives (i.e. high \ICEnvy{} but low \ICRegret) are acceptable while false negatives are not. Therefore the inequality goes in the right direction.} In particular, for the position auction environment \cite{EOS07,V07}, we show that a large class of payment rules (which includes VCG and GSP), $\ICEnvy \ge \ICRegret$, and $\ICEnvy = \ICRegret$ when bids are unique (Section~\ref{s:truthful}). We extend these results to the Ad Types setting \cite{CMSW19} (in which ads of different types have different discount curves\footnote{This is common for example with heterogeneous ad types in an ad auction: the probability of an impression ad being seen decays differently than the probability of a video ad being watched, or the probability of a link-click ad being clicked. All discount curves agree on the relative quality of the slots.}) and show that VCG and a generalization of GSP also have $\ICEnvy \ge \ICRegret$ (Section~\ref{s:egsp}). In addition to bouding \ICRegret{} in terms of \ICEnvy{}, we also use \ICEnvy{} to bound the loss in social welfare loss due to misreports (Section~\ref{s:swl}), and finally we show empirically that \ICEnvy{} can be used as a feature in an estimator for \ICRegret{} in auction environments beyond those for which we have theoretical results (Section~\ref{s:feature})).

\subsection{Related Work}


We propose to connect \ICEnvy{} and \ICRegret{} directly by defining a large class of auction mechanisms for which $\ICEnvy= \ICRegret$ (and a larger class where $\ICEnvy\ge \ICRegret$). The line of work that's closest in spirit aims to identify classes of auction mechanisms that are simultaneously envy-free and IC (in our notation: classes for which $\ICRegret= \ICEnvy = 0$). Feldman et al. \cite{FL12} and Goldberg et al. \cite{GH03} studied the conditions that are required in order to have mechanisms that are efficient, truthful and envy-free and that VCG satisfy these properties for capacitated valuation functions. For homogeneous capacities there's a class of mechanisms that achieve this, while for heterogeneous capacities there is no mechanisms that simultaneously achieved all 3 conditions. Cohen et al. \cite{CFFKO10} provided a characterization based on cycle-monotonicity of the allocation functions that are  incentive-compatible and envy free without considering the efficiency of the algorithms. 

The notion of envy-freeness was initially introduced by Varian \cite{varian1973equity} and Foley \cite{foley1967resource}. The key property of an envy-free allocation is that buyers prefer the bundle of goods they receive over any other allocated bundles (given bundle prices). The notion is particularly appealing due to its connection to markets: in an envy-free allocation, given the prices for goods, all buyers prefer to buy the bundle that's assigned to them. More recently, the notion of envy-freeness has been deeply studied with a different perspective that involves item-pricing \cite{GHKKKM05,DBLP:conf/wine/Colini-Baldeschi16} and bundle-pricing \cite{FW09,FFLS12,DBLP:conf/wine/Colini-BaldeschiLSZ14}. In our setting there is no difference between those two models. A similar line of work focused in studying envy-free algorithms, both in the item-pricing and the bundle pricing models, when the bidders have budget constraints \cite{FFLS12,DBLP:conf/wine/Colini-Baldeschi16,DBLP:conf/wine/Colini-BaldeschiLSZ14,DBLP:conf/mfcs/BranzeiFMZ17,TZ15}.

In much of the other related work, envy-freeness is taken as an alternative solution concept to IC (e.g. \cite{HY11,DS16,CLSZ14,FFLS12,CLZ16,FW09}) and in contexts outside of the auction domain (e.g. \cite{christodoulou2011global,fleischer2011lower,CFFKO10b,MS18}). Of particular note: Daskalakis and Syrgkanis \cite{DS16} address the relation between envy-freeness and incentive-compatibility in the context of algorithmic learning. In particular, the authors discussed the computational complexity of no-regret learning algorithms and no-envy algorithms in simultaneous second price auctions. Hartline and Yan \cite{HY11} studied the relation between envy-freeness and incentive compatibility in revenue-maximizing prior-free mechanisms. Lipton et al.  \cite{lipton2004approximately} investigated envy-free mechanisms in the context of indivisible items with focus on the computational complexity of finding allocations with minimum envy. Moreover, they proved that is possible to obtain truthful mechanisms with bounded envy. Those results have been simplified and extended by Caragiannis et al. \cite{CKKK09}. While this line of work is interesting, it does not quantitatively address the relationship of envy and IC regret.

\subsection{Our Contributions}
This paper has 4 main contributions:
\begin{enumerate}
	\item First, in Section~\ref{s:truthful}, we define a class of auction mechanisms---which includes VCG, GSP, and GFP for position auctions---where  \ICEnvy{} is tightly related to \ICRegret{}. For this class we give necessary and sufficient conditions for $\ICEnvy \ge \ICRegret$ and mild supplementary conditions under which they are exactly equal.
	\item Secondly, in Section~\ref{s:egsp} we consider the more general Ad Types auction environment \cite{CMSW19} in which different ads have different discount curves. We show that for VCG and a suitable generalization of GSP it still holds that $\ICEnvy \ge \ICRegret$.
	\item Third, in Section~\ref{s:swl}, we upperbound the social welfare loss in terms of \ICEnvy{} for the same sets of mechanisms introduced in Sections~\ref{s:truthful} and \ref{s:egsp}. We show that in equilibrium, the social welfare loss is \emph{at most} $4 \cdot \ICEnvy$ (under a technical condition we introduce in the section).
	\item Finally, in Section~\ref{s:feature}, we use bidding data from a major online publisher to show that \ICEnvy{} can be used as a feature to learn an estimator for \ICRegret. The estimator has low mean-squared error, and performs better than comparable estimators that are trained using other features from the auction like values and prices for different slots.
\end{enumerate}

%% file: preliminaries.tex
\section{Preliminaries}
\label{s:preliminaries}

There are $n$ bidders and $m$ slots. Let $I$ be the set of bidders and $J$ be an (ordered) set of slots.  
Each bidder $i \in I$ has a valuation vector $\bm{v_i} = \langle v_{i,1}, v_{i, 2}, \ldots, v_{i,m}\rangle$ that is the willingness to pay of bidder $i$ for each slot $j$, with $ v_{i,1} \ge v_{i, 2} \ge \ldots \ge v_{i,m}$, and are unit demand.  In the standard position auction environment, slots have common quality factor $\alpha_1\ge \alpha_2\ge \ldots\ge \alpha_m$ such that for each bidder $i$ and slot $j$ we have $v_{i,j} = v_i\cdot \alpha_j$ for private value $v_i\in\RR$ of the bidder. In the Ad Types setting, each ad has a type $\theta$ and for each type there is a separate discount curve $\alpha_{\theta,1} \ge \alpha_{\theta,1} \ge ... \ge \alpha_{\theta,m}$. Note that when there's is only a single type, the setting specializes to the position auction environment. Unless specified differently, let $\theta_i$ refer to the type of ad $i$.

The slots are allocated to the bidders by a (direct-revelation) mechanism $\mathcal{M}$. The mechanism $\mathcal{M}$ is defined by an allocation function $\mathcal{A} : \RR^n \to \NN^{n}$ and a payment function $\mathcal{P} : \RR^n \to \RR^n$. Since bidders' values $v_i$ are private, the mechanism solicits bids $b_i$ to represent the values, though reports may not be truthful. Let $\bm{v}$ be the valuation vector of all the bidders and  $\bm{b}$ the bid vector. After receiving the bids from all the bidders, the mechanism $\mathcal{M} = \langle \mathcal{A}, \mathcal{P} \rangle$ computes an outcome $\outcome$, i.e., $\mathcal{A}(\bm{b}) = \allocation$ and $\mathcal{P}(\bm{b}) = \payment$.

$\allocation$ describes the allocation of the slots to the bidders and $\payment$ describes how much each bidder is charged for the obtained slot.
In particular, $\allocation = \langle X_1, X_2, \ldots, X_n \rangle$ where $X_i = j$, if the bidder $i$ obtains the slot $j$ and $0$ if she does not receive any slot. And $\payment= \langle p_1, p_2, \ldots, p_n\rangle$ where $p_i \in \mathbb{R}_{\geq 0}$ is the price that the bidder $i$ pays for slot $X_i$.

For an allocation $\allocation$ and a valuation vector $\bm{v}$, the social welfare of the allocation $\allocation$ is $SW (\bm{v}, \allocation) = \sum_{i \in I}  v_{i} \alpha_{\theta_i, X_i}$.  The optimal social welfare is $SW^{OPT} (\bm{v}) = {\rm max}_{\allocation}\sum_{i \in I}  v_{i} \alpha_{\theta_i,X_i}$. The { Social Welfare Loss } is $SWL (\bm{v}, \allocation) = SW^{OPT} (\bm{v}) - SW  (\bm{v}, \allocation)$. When the valuation vector $\bm{v}$ is clear from the context, we will use $SW (\allocation)$, $SW^{OPT}$, and $SWL (\allocation) $. When the mechanism $\mathcal{M}$ and the truthful valuation vector $\bm{v}$ is clear from context to we use $SW ({\bf b})$, $SW^{OPT}$, and $SWL ({\bf b})$ with the understanding that $\allocation = \mathcal{A}(\bf b)$.

Given an outcome $\outcome$, the utility of a bidder $i$ with type $\theta_i$ is $u_i(X_i, p_i) = v_{i} \alpha_{\theta_i, X_i} - p_i$. Since the outcome of a mechanism $\mathcal{M}$ is a function of the bids, and the auctions we consider are not necessarily IC, bidders may be incentivized to report a type $\bm{b}$ different from $\bm{v}$ in order to produce an outcome with higher utility. 

\paragraph{IC-Regret.}  \ICRegret{} describes the outcome for bidding truthfully, compared to the optimal alternative bid (given constant competition $b_{-i}$). Formally, the regret of a bidder $i$ for bidding truthfully compared to a \emph{specific} alternative bid $b_i$ is:
\begin{equation}
r_i(b_i,b_{-i},v_i) = \max_{b_i \in \mathbb{R}^+} \{0, u_i(\mathcal{A}(b_i,b_{-i}),\mathcal{P}(b_i,b_{-i}) ) -
u_i(\mathcal{A}(v_i,b_{-i}),\mathcal{P}(v_i,b_{-i}) ) \},
\end{equation}
which is used in the formal definition for \ICRegret{}.
\begin{definition}[\ICRegret]
	The \ICRegret{} that bidder $i$ experiences is\footnote{Equation~\eqref{eq:ic-regret} in the introduction takes an expectation over competition $b_{-i}$ since the work of Feng et al.~\cite{FSS19} considers the auction mechanism as a black box. In our setting (from the perspective of the auctioneer) the alternative bids are known, and we define \ICRegret{} on an auction-by-auction basis.}
	\begin{equation*}
	\ICRegret_i(v_i, b_{-i}) = \max_{b_i\in \RR_{\geq 0}} \{r_i(b_i,b_{-i}, v_i)  \}.
	\end{equation*}
\end{definition}

\ICRegret{} can be directly connected to {\it incentive-compatibility (IC)}. Indeed, a mechanism $\mathcal{M}$ is IC iff for all $v_i$, $b_{-i}$, and  $i \in I$, we have $r_i(b_i,b_{-i}, v_i) = 0$ for all $b_i \in \RR_{\geq 0}$.

\paragraph{IC-Envy.} Given an allocation $\allocation$ and payments $\payment$, Envy describes how much a bidder prefers the allocation and price of another buyer, compared to what they received themselves. Since different bids may lead to a different auction outcome, we define \ICEnvy{} as Envy with respect to the allocation $\allocation$ and payments $\payment$ when bidder $i$ bids \emph{truthfully}. \ICEnvy{} is some notion of fairness of the produced outcome whereas \ICRegret{} measures how much the underlying mechanism incentives misreported types.

Formally, for given an allocation $\allocation$ and payments $\payment$, the envy that bidder $i$ experiences compared to bidder $j$ is
\begin{equation}
e_i^j(\allocation, \payment) = \max\{0, u_i(X_j,p_j) - u_i(X_i, p_i) \},
\end{equation}
and the envy of bidder $i$ in the outcome $\outcome$ is
\begin{equation}
E_i(\allocation, \payment) = \max_{j \in I \setminus \{i\}}\{e_i^j(\allocation, \payment) \}
\end{equation}
which is used in the formal definition for \ICEnvy{}.

\begin{definition}[\ICEnvy{}]
	\ICEnvy{} the envy of bidder $i$ in the outcome $\langle\mathcal{A}(v_i, b_{-i}), \mathcal{P}(v_i, b_{-i})\rangle$:
\begin{equation*}
\ICEnvy_i(v_i, b_{-i}) = \max_{j \in I \setminus \{i\}}\{e_i^j(\mathcal{A}(v_i, b_{-i}), \mathcal{P}(v_i, b_{-i})) \}.
\end{equation*}
\end{definition}

Note that computing \ICEnvy{} requires one single execution of the auction whereas the computation of \ICRegret{} requires the execution of the auction for multiple bid values of each bidder.

%% file: truthful.tex
\section{Position Auction Environments}
\label{s:truthful}

As stated before,  \ICEnvy{} and \ICRegret{} measure different things: \ICEnvy{} provides some measure of fairness of the outcome, whereas \ICRegret{} measures the incentive-compatibility of the mechanism. In this section, we focus on position auctions that are widely used in search and feed advertising.\footnote{In display advertising it is more common to sell ad slots one-by-one, which is a special case of position auctions, though one which is arguably mathematically less interesting.} We give in the following the definition of {\em regular mechanism} for position auctions and we characterize the class of regular mechanisms that have  $\ICEnvy{} \ge \ICRegret{}$.  We assume wlog that the bidders are ordered by non-increasing bid $b_i$, with ties broken lexicographically. Therefore, slot $i$ is assigned to bidder $i$.  
\begin{definition}[Regular Mechanisms for  Position Auctions]
	\label{def:regular}
	A regular mechanism $\mathcal{M}$ for position auctions is defined as follows:
	
	\begin{enumerate}
		\item  Slots are assigned  in order of non-increasing $\alpha_i$ to bidders ordered by non-increasing bid value $b_i$.  Ties are broken lexicographically. 
		\item The payment for bidder $i$ is $p_i =  \sum_{k=1}^n a_{i,k} \cdot b_{k}$ with non negative coefficients $a_{i,k}\geq 0.$
	\end{enumerate}
\end{definition}
\noindent Note that this  definition includes several widely used auction mechanisms: 
\begin{itemize}
\item VCG:  $a_{i,k}= 0$ for $k=1,\ldots,i$, and  $a_{i,k}=\alpha_{k-1}-\alpha_{k}$ for $k=i+1,\ldots, n$.
\item GSP:  $a_{i,k}= 0$ for $k=1,\ldots,i$,  $a_{i, i+1}=\alpha_{k}$, and $a_{i,k}=0$ for $k=i+2,\ldots, n$.
\item GFP: $a_{i,k}= 0$ for $k=1,\ldots,i-1$, $a_{i,i}=\alpha_{k}$, and $a_{i,k}=0$ for $k=i+1,\ldots, n$.
\end{itemize}

In this section we provide necessary and sufficient conditions for a regular mechanism to be individually rational, i.e., no bidder is charged more than her bid, and to have for each bidder $i$, $\ICEnvy_i(v_i) \geq \ICRegret_i(v_i).$ 

\begin{lemma}
\label{thm:iff}
For a regular mechanism for position auctions the following properties hold
\begin{itemize}
\item[i] Individual Rationality;
\item [ii]  $\ICEnvy (v_i) \geq \ICRegret (v_i)$,
\end{itemize}  
if and only if, for each slot $i$, 
\begin{itemize}
\item $a_{i,k}=0, k=1,\ldots, i$, and 
\item $ p_i - p_{i+1}\geq  (\alpha_{i}- \alpha_{i+1})  b_{i+1}$ for $i=1, \ldots, n-1$. 
\end{itemize}
\end{lemma}
\begin{proof}
We start by proving the necessity of the first condition of the claim. Conditions  $a_{i,k}=0$, $k=1,\ldots, i-1$ are needed to ensure the individual rationality of the mechanism. Indeed, if there exists a coefficient $a_{i,k} >0$, $k<i$, any bidder $i$ with  $v_i<a_{i,k} b_k$ will be charged more than its valuation, thus violating individual rationality. 
Condition  $a_{i,i}=0$ is also needed. Otherwise, if  $a_{i,i}>0$, bidder $i$ may have $r_i(v_i-\epsilon ,b_{-i},v_i)>0$ while $\ICEnvy(v_i)=0$ as shown in Example~\ref{ex:fp}. 

Next, we prove that the second condition of the claim is  sufficient. We first prove that regret and envy are $0$  for all slots $j < i$.  For envy, we derive: 
\begin{align*}
e_i^j(\allocation, \payment) &= \max\{0, u_i(X_j,p_j) - u_i(X_i, p_i) \} \\
					  &\le  \max\{0, (\alpha_j - \alpha_i) v_i - (p_j - p_i) \}\\
					  &\leq 0,
\end{align*}
with the last inequality obtained by the second condition on the payments.  For regret, note that for a bid $b_i'$ such that $X_i=j$, $j< i$, we have a payment $p_j(b_i', b_{-i}) \geq p_j(v_i, b_{-i})$ and therefore $r_i(b_i', b_{-i}, v_i) \leq e_i^j \leq 0$. 
 
Envy and regret  for bidder $i$ can only be positive for a slot $j > i$.   if slot $j$ can be obtained from bidder $i$ by decreasing her bid to a value $b_i'$, then, given the first condition of the theorem, the payment charged to agent $i$ for a bid $b_i'$   that gives him  slot $j \geq i$ is exactly equal to the payment  charged to the bidder that received slot $j$ under bid vector $b=(v_i,b_{-i})$, namely, $p_j(b_i', b_{-i}) = p_j(v_i, b_{-i})$. 
The reason is that,  all the bids $b_{j+1},\ldots, b_n$ that determine the payment are unchanged. Therefore,  $e_i^j(\allocation, \payment) = r_i(b_i', b_{-i}, v_i)$. 

However, not all slots $j>i$ can be obtained from bidder $i$ by decreasing her bid since ties are broken lexicographically.  If there exists a slot $j>i$ that cannot be obtained from bidder $i$ for any bid $b_i'<v_i$, then we have still  the possibility that $e_i^j(\allocation, \payment) > R_i(v_i, b_{-i}$, and therefore  $\ICEnvy (v_i) \geq \ICRegret (v_i)$. 
\end{proof}

The part where the strict inequality $\ICEnvy (v_i) > \ICRegret (v_i)$ came in was due to the lexicographic tie-breaking when there are ties. When bids are distinct this case disappears and  $\ICEnvy (v_i) = \ICRegret (v_i)$.

\begin{theorem}
\label{coro:all-slots}
When all bids $b_1, \ldots, b_n$ are different,  the conditions of Lemma~\ref{thm:iff} are necessary and sufficient for individual rationality and  $\ICRegret(v_i) =  \ICEnvy(v_i).$
\end{theorem}
\begin{proof}
In addition to the necessary conditions of Lemma~\ref{thm:iff} we prove that  $p_i - p_{i+1}\geq  (\alpha_{i}- \alpha_{i+1})  b_{i+1}$ is also necessary condition for $\ICRegret(v_i) =  \ICEnvy(v_i)$. By contradiction, consider the largest value of $i$ such that  $p_i - p_{i+1}< (\alpha_{i}- \alpha_{i+1})  b_{i+1}$. Bidder $i+1$ will envy the allocation of bidder $i$ since the marginal utility of having allocated slot $i$ instead of slot $i+1$ is positive.  On the other hand, increasing the bid of bidder $i+1$ to obtain slot $i$ will have negative regret for bidder $i+1$ if bid $b_i$ is enough bigger than bid $b_{i+1}$ and therefore the payment of $p_i$ will be increased more than the marginal utility of bidder $i+1$ for slot $i$.  

For the sufficient condition,  note that, given the fact that all bids are different, we never have ties in the allocation of a slot. Therefore, for every slot $j>X_i$,  there  exists a bid $b'_i$ such that $\mathcal{A}(b', b_{-i})$ produces $X_i = j,$ and the payment is exactly equal to the payment  charged to the bidder that received slot $j$ under bid vector $b=(v_i,b_{-i})$. Therefore $\ICRegret(v_i) =  \ICEnvy(v_i).$
\end{proof}

Lemma~\ref{thm:iff} and Theorem~\ref{coro:all-slots}  hold for mechanisms like VCG, GSP, and any combination of the two.  The first condition is clearly true for the two mechanisms.  The second condition is true with $p_i - p_{i+1}= (\alpha_{i}- \alpha_{i+1})  b_{i+1}$ for VCG and with $p_i - p_{i+1}=  \alpha_{i} b_i - \alpha_{i+1}  b_{i+1} \geq (\alpha_{i}- \alpha_{i+1})  b_{i+1}$ for GSP.  On the contrary, even the first condition of the theorem is violated for GFP.




%

%% file: extendedGSP-monotonic.tex
\section{The Ad Types Environment}
\label{s:egsp}
We now consider a more general setting that position auction, in which bidders can have different types and thus face different discount curves, as proposed by Colini-Baldeschi et al. \cite{CMSW19}.  
Each bidder $i\in I$  is associated with a  vector $\alpha_{i,1}, \ldots, \alpha_{i,m}$  of non-increasing  quality values for the $m$ slots. 
The valuation of bidder $i$ on slot $j$ is  $v_{i,j}= v_i  \cdot \alpha_{i,j}$.  Valuation $v_i$ remains private information of the bidders, while the mechanism knows the quality value vectors ${\bf \alpha}_i$ for each bidder $i$. 

The social-welfare maximizing allocation is obtained by solving the Max-Weight Perfect Matching (MWPM) problem, 
for example by using the Hungarian method. This returns a perfect matching $M$ together with a dual certificate of its optimality.  
The certificate is a dual price vector $\bf p$ for the $m$ slots $J$ and a dual utility vector $\bf q$ for the $n$ bidders $I$, 
such that the value of the optimal solution is equal to $\sum_i p_i + \sum_j q_j$, i.e., the sum of the prices of the slots plus the sum of the utilities of the bidders. If bidder $i$ is matched to slot $j$, the dual constraint $p_j + q_i \geq  v_{i,j}$ holds with equality.  
Let $E^=$ be the set of tight edges in the final solution. The MWPW $M$ is therefore  a subset of  $E^=$. 

In the exposition, we assume that the number of ads $n$ is equal to the number of slots $m$ and that the matching $M$ is unique 
for each instance of the problem. This can be achieved by first adding slots if $n>m$ that each bidder values at $0$ or by removing 
the lowest slots if $n<m$, followed by a deterministic perturbation of position discounts to remove ties on the value of any subset of edges. A formal description of this process 
is given in the appendix (see Appendix~\ref{ss:ties}).

The following properties hold for shadow prices when using a suitable variation of the Hungarian method 
\cite{CMSW19, DBLP:journals/corr/KernMU16, CMSW19}: 

\begin{enumerate}
\item Wlog, the Hungarian algorithm can output dual prices $\bf p$ that are pointwise minimal over feasible dual solutions of the max-weight allocation \cite{CMSW19}.
\item For the minimal dual prices, prices $p_j$, $j=1,\ldots,m$  of the slots are non increasing, and the final slot is free: $p_m = 0$.
\item For the minimal dual prices, each bidder $j$ is connected with an alternating path $P_j \subseteq E^=$ to an item $j$ with price $p_j=0$ \cite{DBLP:journals/corr/KernMU16}.
\item The minimum dual prices are market clearing prices, i.e. each bidder is matched with a slot that maximizes their utility.
\end{enumerate}

In addition to the properties above, the following is true (the proof appears in the appendix).

\begin{claim} 
\label{claim:low_tight}
Each bidder $i$ is matched to the \emph{lowest} slot $j$ for which $(i,j) \in E^{=}$. 
\end{claim}

\subsection{The Extended GSP Pricing Scheme}
We want to handle pricing schemes for this setting with different discount curves for different ads, but what pricing schemes should be considered? We focus on attention on pricing schemes with the following fundamental properties:

\begin{enumerate}
\item Prices are monotonically non increasing, i.e.,  for each bidder $i$, slots of lower quality do not have higher price;
\item $\ICEnvy \leq \ICRegret$
\end{enumerate}

The first property is an obvious requirement since the discount curves of all the bidders are non-increasing. 
Observe that in this setting, VCG prices satisfy both constraints trivially since $\ICEnvy = \ICRegret = 0$, but what about other pricing rules 
that charge higher prices? 
GSP for this setting can be generalized, e.g. as done in \cite{CW14,CSW18}, by considering \emph{charging the value for a slot corresponding to the lowest bid that maintains the same allocation}. For position auctions, this specializes to the normal GSP pricing scheme. The following example shows that this pricing scheme fails to preserve price monotonicity:

\begin{example}
Consider three bidders with valuations over three slots $(10,9,8), (7,6,4) ( 4,0,0)$. The threshold prices of the three slots are $p_1=2, p_2=3, p_3=0$.  For the first slot, we need to reduce the valuation of bidder $3$ to 2 in order to have the first two slots assigned to bidders $1$ and $2$ for a total value of $16$. For the second slot, if we scale down by $1/2$ the bid of bidder $2$ (we are in a single parameter setting) and we bring the values of bidder $2$ to $(3.5, 3,2)$, it is convenient to assign slot $2$ to bidder $1$ instead of bidder $2$ that is now assigned to slot $3$. 	
\end{example}

\paragraph{Extended GSP.} The \emph{extended GSP} pricing scheme prices slots based on the the analysis of the set  $E^=$ of tight edges connecting a bidder $i$ with the slot of price $0$.  Indeed, the standard GSP mechanism  with bidders of only one ad type prices each slot $i$ at the value of bidder $i+1$ for slot $i$, and the highest value of  a tight edge on the alternating path from bidder $i$  to slot $m$ with price $p_m=0$.  

However, we cannot extend directly this mechanism  to the case of different ad types. 
We show a simple example in which the final prices for the case of three different ad types cannot be 
higher than the VCG prices. 

\begin{example}
Consider three bidders with valuations over three slots $(10,9,8), (7,6,4) ( 4,0,0)$. 
The optimal VCG solution will match the bidders according to the order $3,2,1$ and will charge 
$p_1=2, p_2=1, p_3=0$.  We also observe that bidder $1$ is tight with slots $1$ and $2$.  However, the values of bidder $1$ for slots 
$1$ and $2$ are  higher than the values of bidders $3$ and $2$, respectively.  Bidder $2$ is also tight with slot $1$ but his value is higher than the value of bidder $3$ that has assigned slot $1$.  We must therefore conclude that the largest price we can charge for the slots are exactly equal to the VCG prices. 
\end{example}

The example above suggests to limit the maximum price that a slot can be charged. 
Consider a bid vector $\bf b$ for the $n$ bidders. 

\begin{definition}{\bf Extended GSP prices.}
\label{def:gsp-prices}
Consider bidder $i$ assigned to slot $j$ with VCG clearing price $p_j$.  Let $E^=_j$ 
be the set of tight edges $(i',j')$, with $j'\geq  j$ and bidder $i'$ matched with an item $j'\geq j$.  
Let $ t_{i,j} = {\tt max}_{(i',j')\in E^=_j} \{b_{i',j'}:  b_{i',j'} <\leq b_{i,j}\},$ i.e., 
the maximum value smaller than $b_{i,j}$ of an edge in $E^=_j$.  
We assume $t_{i,j}=0$ if no edge  in $E^=_j$ has value  smaller than $b_{i,j}$.
We define the extended GSP price of bidder $i$ for slot $j$ by 

\begin{eqnarray}
\label{def:gsp-prices}
gsp_{i,j}= {\tt max} \{p_j,  t_{i,j}\}
\end{eqnarray}

\end{definition}
Observe that the prices are non-anonymous since we must impose a different upper bound on the maximum value that we can charge each bidder for a given slot. 
Indeed, by similar arguments used for GSP in Section \ref{s:truthful}, these are intuitively the largest prices we can charge a bidder in a setting with different ad types without violating the property $\ICEnvy \geq \ICRegret$. Any larger value of the payment of bidder $i$ for item $j$ can only depend 
on the bid of bidder $i$. If we had this dependence, bidder $i$ can decrease his payment by decreasing the bid and therefore $\ICRegret$ would be higher than $\ICEnvy$. 

We first observe the following: 

\begin{claim}
The extended VCG prices $\gsp_{i,j}$ are monotonic, i.e., $\gsp_{i,1}\geq \ldots \geq \gsp_{i,m}$.
\end{claim}
\begin{proof}
Monotonicity follows directly from the definition of extend VCG prices by observing that VCG clearing prices $p_j$ are non increasing, $E^=_j \subseteq E^=_{j-1}$, and, finally, $\bf \alpha_{i}$ is a vector of non increasing quality values.  
\end{proof}

We also observe that the extended GSP prices are not monotonic. An open problem that we pose is the one of finding anonymous monotonic prices for different ad types such that  $\ICEnvy \geq \ICRegret$. 



\subsection{$\ICEnvy \geq \ICRegret$ for Extended GSP}
In the following, we prove for extended  $GSP$  that $\ICEnvy \geq \ICRegret$. 
In order to prove this result, we need to argue as follows.  Assume bidder $i$ is matched to slot $j$ at price $\gsp_{i,j}$ with truthful bid $b_i=v_i$, and let us also assume there is envy for slot $\tilde j$ at price $\gsp_{i,\tilde j}$.  If  bidder $i$ modifies the bid to $\tilde b_i$ in order to be matched to $\tilde j$, then we have  $\tilde \gsp_{i,\tilde j} \geq gsp_{i,\tilde j}$, and therefore $\ICEnvy$ is at least as large as  $\ICRegret$. We use in the following the simple fact that a slot of higher quality value can only be obtained by increasing $b_i$, and, symmetrically, a slot of lower quality value can only be obtained by decreasing $b_i$.

Let us start by proving that the set of tight edges $E^=_j$ can only be larger if   bid $b_i$ is increased (proof appears in the appendix).  

\begin{claim}
\label{claim:price_increase}
Let  $b_i \leq \tilde b_i$, and let $j$ and $\tilde j \leq j$  the slots assigned to bidder $i$ with bids $b_i$ and $\tilde b_i$, respectively.
Then, $E^=_{\tilde j} \subseteq \tilde E^=_{\tilde j}$ and $\tilde p_{\tilde j} \geq p_{\tilde j}$.
\end{claim}

The second claim considers the decrease of bid $b_i$. 

\begin{claim}
\label{claim:price_decrease}
Let  $\tilde b_i \leq b_i$ and let $\tilde j$ and $j \leq \tilde j$,  the items assigned to bidder $i$ with bids $\tilde b_i$ and $b_i$, respectively. Then, 
$\tilde E^=_{\tilde j}= E^=_{\tilde j}$ and $p_{\tilde j} = \tilde p_{\tilde j}$. 
\end{claim}

We then argue about the relation between $\ICRegret$ and $\ICEnvy$ when the bid $b_i$ of bidder $i$ is modified. 

\begin{lemma}
Let  $b_i \leq \tilde b_i$, and let $j$ and $\tilde j \leq j$  the slots assigned to bidder $i$ with bids $b_i$ and $\tilde b_i$, respectively.
Then  $\tilde gsp(i, \tilde j) \geq gsp(i, \tilde j)$.
\end{lemma}
\begin{proof}
Claim \ref{claim:price_increase} shows that the increase of  bid $b_{i}$ yields  $E^=_{\tilde j} \subseteq \tilde E^=_{\tilde j}$ and $\tilde p_{\tilde j} \geq p_{\tilde j}$.  Then, the price $\gsp (i,\tilde j)$ can only increase if  $b_i$  is increased. 
\end{proof}

\begin{lemma}
Let  $\tilde b_i \leq b_i$, and let $\tilde j$ and $\tilde j \geq j$  the items assigned to bidder $i$ with bids $\tilde b_i$ and $b_i$, respectively.
Then $\tilde gsp(i, \tilde j) = gsp(i, \tilde j)$.
\end{lemma}
\begin{proof}
Claim \ref{claim:price_decrease} implies that $\tilde E^=_{\tilde j}= E^=_{\tilde j}$ and $p_{\tilde j} = \tilde p_{\tilde j}$. We therefore conclude that price $\gsp (i,\tilde j)$ is not reduced if  $b_i$  is decreased.
\end{proof}

We therefore conclude with the following: 

\begin{theorem}
For extended $GSP$ it holds $\ICRegret \leq ICEnvy$. 
\end{theorem}

%% file: swl-Matching.tex
\section{Measuring Social Welfare Loss with \ICEnvy{}}
\label{s:swl}
We proved in the previous sections that \ICEnvy{} is an upper bound on \ICRegret{} for a large class of mechanisms and that they are equal under mild conditions on the bid vector. 
%
%
%
%
%

We next show that \ICEnvy{} can also be used to measure the efficiency of the auction.  We show a direct connection between {\em Social Welfare Loss (SWL)} and the \ICEnvy{}  experienced by all the bidders. This connection will be proved under the assumption that any bidder $i$ is bidding a bid $b_i$ that gives a utility not lower than the truthful strategy, namely $u_i(b_i, b_{-i}) \geq u_i(v_i, b_{-i})$. Observe that we do not require the bidders to play a min-regret strategy, neither we assume the bid vector $\bm{b}$ to be a Nash Equilibrium.
Indeed, assuming that the bidders report a bid that is at least as good as their truthful valuations is not too restrictive, since reporting the truthful value is always a feasible option, whereas computing the min-regret strategy is computationally expensive. 

The connection between SWL and \ICEnvy{} is proved using the notion of {\em smoothness} which has been introduced in \cite{Roughgarden09}. This concept was introduced to prove bounds on the price of anarchy \cite{Roughgarden2017} of an auction at the equilibrium. We use a relaxation of the notion of smoothness called {\em semi-smoothness}. Semi-smoothness has been introduced \cite{LP11} with the goal of studying the efficiency of a position auctions even off equilibrium. 

The notion of semi-smoothness is defined as follows: Given a  bid vector $\bm{b}$, there exists an alternative bid vector  $\bm{b}'$ such that
$$\sum_{i \in I} u_i(b_i, \bm{b}_{-i})\geq \lambda SW(\bm{b}')-\mu SW(\bm{b})$$ for suitable constants $\lambda$ and $\mu$.  In our specific case we have $\lambda=1/2$ and $\mu=1$. The state $\bm{b}'$ is actually obtained by setting for each bidder $b'_i = v_i/2$.

In order to prove our result, we extend the following claim proved in \cite{LP11} for GSP to the extended GSP. We remind to the reader that positions are ordered by non increasing quality value $\alpha_j$.  We denote by $\pi(\bm{b},j)$ the advertiser allocated to slot $j$ under bid vector $\bm{b}$. We denote by $X(\bm{b},i)$ the slot allocated to bidder $i$ under bid vector $\bm{b}$.
Moreover, Let $R_i(b_i,\bm{b}_{-i} )$ be the maximum regret that the bidder $i$ can experience with respect to the bid $b_i$, i.e., 
\begin{equation*}
	R_i(b_i, b_{-i}) = \max_{b'_i\in \RR_{\geq 0}} \{r_i(b'_i,b_{-i}, b_i)  \}.
\end{equation*}
All proofs of this section appear in the appendix.
\begin{claim}
\label{cl:1}
	Fix a valuation profile $\bm{v}$ and an agent $i$. Let us denote $j=X(\bm{v},i)$ the optimal assignment to bidder $i$ with truthful bids. Consider any bid profile $\bm{b}$ an define $b_i'=v_i/2$.
	We claim that  $$u_i(v_i/2, \bm{b}_{-i}) + \alpha_{\pi(\bm{b},j),j} v_{\pi(\bm{b},j)} \geq 1/2 \alpha_{i,X(\bm{v},i)} v_i$$
\end{claim}
%
%
%
%

%
%
%
%

By applying the claim above we derive the following: 

\begin{eqnarray}
\sum_i R_i(b_i,\bm{b}_{-i} ) &\geq& \sum_i  u_i(b'_i, \bm{b}_{-i}) - \sum_i \alpha_{i,X(\bm{b},i)} v_i \nonumber\\
&\geq& \sum_i 1/2 \alpha_{i,X(\bm{v},i)} v_i - \sum_i  \alpha_{\pi((v_i/2,\bm{b}_{-i}),j),j} v_{\pi(\bm{b},j)} - \sum_i \alpha_{i,X(\bm{b},i)} v_i \nonumber\\
&\geq& \sum_i 1/2 \alpha_{i,X(\bm{v},i)} v_i - \sum_i  \alpha_{\pi(\bm{b},j),j} v_{\pi(\bm{b},j)} - \sum_i \alpha_{i,X(\bm{b},i)} v_i \nonumber\\
&\geq& \frac12 SW^{OPT} - 2 SW(\bm{b})
\end{eqnarray}

where the first inequality stems from the fact that the maximum regret of bidder $i$ playing $b_i$ is lower bounded by the regret obtained when playing $b'_i$ and the second inequality derives from the fact that the MWPM can only decrease if the valuation of a bidder  decreases. 

In the next claim we show that the regret at $v_i$ is larger than the regret at $b_i$. 

\begin{claim}
\label{claim:regretbound}
	$R_i(b_i, b_{-i})\leq R_i(v_i, b_{-i})$
\end{claim}
%
%
%

We therefore conclude with the following:
\begin{theorem}  
\label{thm:swl}
If $SW^{OPT}(\bm{v}) \geq 8  SW(\bm{b})$ then $\sum_{i}E_i(v_i, \bm{b}_{-i}) \geq \frac{1}{4} SWL(\bm{b})$.
\end{theorem}

%% file: ml-experiments.tex
\section{Using Envy as a Feature}
\label{s:feature}
In Sections~\ref{s:truthful}~and~\ref{s:egsp} we showed that for reasonably large classes of mechanisms, \ICRegret{} can be expressed in terms of \ICEnvy{} and the two quantities are equal when all bids are different. In this section we move beyond linear relationships between envy and regret, and show that using envy as a feature can lead to better ML algorithms. In particular, we'll show that we can predict regret with reasonable accuracy in auction environments that are far more general than the ones discussed on Section~\ref{s:truthful}.

\subsection{Sanity Check}
Before we focus on using \ICEnvy to predict \ICRegret, Figure~\ref{f:sim} shows a sanity check to see if \ICEnvy is a proxy for \ICRegret for auction mechanisms that aren't explicitly covered by Theorem~\ref{thm:iff}. In Figure~\ref{f:sim} we plot \ICEnvy against \ICRegret for approximately 1M auction bids from a major online publisher (collected on February 20, 2019) in approximately 10K auctions. For each set of bids, we simulate a GFP auction using 10 slots with geometric decaying discount curve. We use GFP because it can be expressed as a Regular Mechanism (according to Def~\ref{def:regular}) but it does not satisfy the payment condition in Theorem~\ref{thm:iff}. While \ICEnvy does not equal \ICRegret, it is an upper bound for it, and the bound is reasonably tight.

\begin{figure}
	\centering
	\includegraphics[width=.8\textwidth]{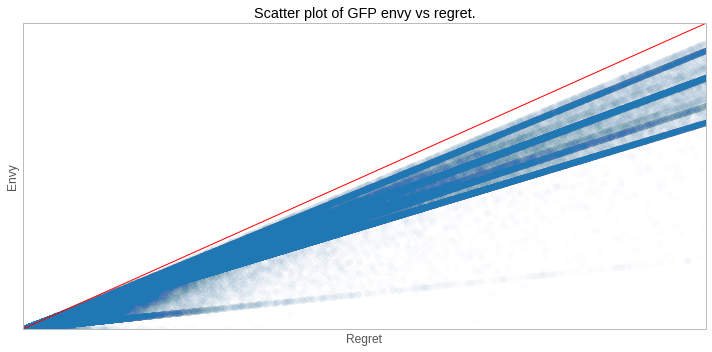}
	\caption{\ICEnvy{} plotted against \ICRegret}
	\label{f:sim}
\end{figure}

\subsection{Using Envy to Predict Regret}
We now go beyond the linear relationship of envy and regret and will use the former to predict the latter.
\subsubsection{Experimental Setup}
{\it Auction Environment.} We look at auctions with $5$ slots, where different bidders have different monotonically decreasing discount curves over the slots (cf. the Ad Types model in Section~\ref{s:egsp}). In this setting, not all ads can target all slots (as a consequence of the different discount curves, not by assumption), and the greedy allocation algorithm is no longer optimal. The auction mechanism that we consider use the greedy allocation (for each slot from highest to lowest, assign the slot to the unassigned ad with the highest discounted value). Using a greedy algorithm instead of the max-weight bipartite matching algorithm means that the theoretical results from Section~\ref{s:egsp} do not apply here. The goal is to show that \ICEnvy{} is useful even outside the setting covered by the theory. We consider 2 pricing rules:
\begin{itemize}
	\item Generalized Second Price (GSP). The discounted value of the next highest bidder, i.e. during the greedy algorithm, the next-highest value ad.
	\item Externality pricing.\footnote{If the allocation algorithm optimized social welfare, then externality pricing would be VCG pricing, and the resulting auction would have $0$ envy and $0$ regret. Since greedy isn't optimal, generally both envy and regret are positive.} The social welfare loss of other buyers due to the presence of buyer $i$.
\end{itemize}

{\it Datasets.} We generate the datasets by drawing bids from a lognormal distribution\footnote{Real-world bids in online auctions typically follow a log-normal distribution, see e.g. \cite{OS11}.} and using 3 classes of bidders with geometric discount curves with parameters $\alpha_1 = 0.9, \alpha_2 = 0.7, \alpha_3 = 0.5$. For each bidder in an auction, a datapoint corresponds to the envy profile (meaning for each of the 5 slots, the unclamped, possibly negative, envy) and the label is the regret.

{\it Baseline.} We compare the performance of the ML models trained on envy, with models that were trained using the (value, price) profile (meaning for each slot, what is the discounted value, and what is the current slot price).

{\it Implementation.} We use scikit-learn \cite{scikit-learn} to train the different models. In particular we use support vector regression (SVR) with the RBF kernel; gradient-boosted regression trees (GBRT) with least-squares loss function, learning rate of $0.1$, and $100$ trees; and neural nets (NN) with 2 hidden layers (of 100, and 20 nodes each) and Adam solver \cite{KB14}.

\subsubsection{Results}
Figure~\ref{fig:mse} shows the  training and cross-validation mean-squared error (MSE) as a function of the number of training samples for the GBDT. The MSE quickly decreases to about $0.02$ after 30K iterations and remains relatively stable after that.


\begin{figure}
\centering
	\includegraphics[width=.6\textwidth]{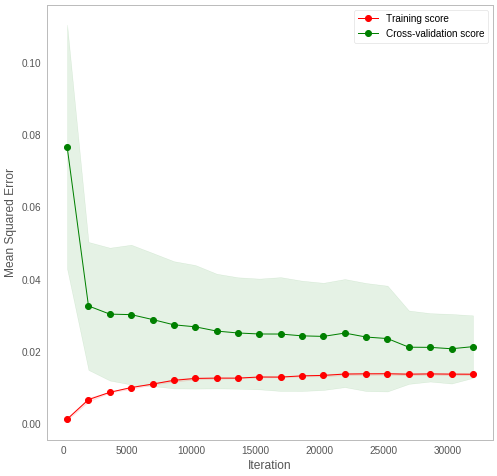}
	\caption{The training and validation MSE of the GBDT model on GSP data as a function of the number data points used to train the model.}
	\label{fig:mse}

\end{figure}

So using envy, we can construct a model that accurately predicts regret. To show that envy uniquely does this compared to reasonable benchmarks, we compare it against models that were trained using price and discounted value for each slot as features; the results are in Table~\ref{t:comparison}. The models here are trained using 100K datapoints, the point being not to train as accurate as possible of a model, but rather to compare the performance of models trained on different features given an equal amount of data. Across all 3 models, the regressor trained on the envy feature alone does better than one that is trained on both the values and prices for slots. This remains qualitatively true for smaller training data sets as well. None of the regressors in the table are necessarily great, but the goal here is not to tweak a regressor to perform well; rather it is to show that using envy as a feature gives better results across a wide variety of models without tuning the model for this particular case.

\begin{table}
	\centering
	\begin{tabular}{|c|c|c|c|}
		\hline 
		& \footnotesize{$R^2$ Price and Value} & \footnotesize{$R^2$ Envy} \\ 
		\hline 
		SVR & $55.9\%$ & $78.4\%$  \\ 
		\hline 
		GBRT &$46.8\%$  & $84.4\%$ \\ 
		\hline 
		NN &  $77.1\%$& $86.2\%$ \\ 
		\hline 
	\end{tabular} 
	\caption{Comparing using price and value as features vs. using envy as features across a range of models trained on 100K datapoints.}
	\label{t:comparison}
\end{table}

%% file: conclusion.tex
\section{Conclusions}
In this paper we proposed to use \ICEnvy{} to give insight in an ad auction in four ways. First, we defined a class of auction mechanisms for position auctions---which includes VCG, GSP, and GFP---where \ICEnvy{} and \ICRegret{} are tightly related. For this class we gave necessary and sufficient conditions for \ICEnvy{} to  upperbound  \ICRegret{} and mild supplementary conditions under which they are exactly equal. 
Secondly, we consider the Ad Types setting, with multiple discount curves, and show that a suitable generalization of GSP (as well as VCG) continue to have $\ICEnvy \ge \ICRegret$.
Thirdly, we upperbounded the social welfare loss in terms of \ICEnvy for the same sets of mechanisms. We show that the social welfare loss is \emph{at most} $4 \cdot \ICEnvy$ (under a technical condition we introduce in the section). 
Finally, we used bidding data from a major online publisher to show that \ICEnvy can be used as a feature to learn an estimator for \ICRegret. The estimator has low MSE, and performs better than comparable estimators that are trained using other features from the auction like values and prices for different slots. 
For future work, we plan to extend our  study of the relationship between \ICEnvy and \ICRegret to the case of bidders with different ad types.  Most importantly, we plan to investigate the existence of a mechanism with monotonic anonymous prices such that $\ICEnvy \geq \ICRegret$. On the more practical side,   we  plan to use  \ICEnvy as a feature to learn an estimator of the social welfare of the auction.

%% file: appendix.tex
\section{Handling Unbalanced Graphs and Ties}
\label{ss:ties}

We conclude the section by showing how to handle unbalanced graphs and ties. If the number of advertisers $n>m$ the number of ties, we add $n-m$ slots, and for all slots $j\in J$ let $\alpha_{i,j} = 0$ for each new slot $j$. Note that we again have a complete bipartite graph, and that the value of the max-weight bipartite matching hasn't changed. If the number of slots $m > n$, we can remove the lowest $m-n$ slots. Due to monotonicity of the discount curves, the lowest $m-n$ slots go unassigned in the max-weight matching, hence they can be safely removed.

Finally, we handle ties in the MWPM (including any ties we introduced by adding new slots).   We do that by 
perturbing by suitable small values the quality values $\alpha_{i,j}$ and then computing the unique MWPM on the perturbed values. 
The returned MWPM is also maximum on the original values. The computed payments differ only by small values from the payment 
computed on the original quality values. However, we obtain the exact same payments on the original quality values by rounding to the closest multiple of the minimum difference between two  $v_i \alpha_{i,j}$ values.  

More formally, let $\delta$  be the minimum difference between two  quality values $\alpha_{i,j}$. Assume the minimum valuation $v_i$ of 
a bidder to be equal to 1. Let us define  $\epsilon = \delta/2^{m^2 +3}$ for a small constant $c$.  We order the quality values $\alpha_{i,j}$ of the slots in any order and we increase the $k$-th quality value in the order by  $\epsilon 2^k$.   This  ensures that any two subsets of edges have different total value and therefore the MWPM is unique.  The payments of the slots are computed as in Definition \ref{def:gsp-prices}.
The VCG price $p_j$ is obtained by subtracting the values of two sets of  edges and therefore the absolute difference of the VCG price for a slot computed on the perturbed values and the VCG price computed on the original values is less than $\epsilon \times 2^{m^2+1}$.  By summing up the perturbed value $t_{i,j}$ of an edge, we increase the difference by at most the largest perturbation. In total, we have a difference that is less than $\epsilon \times 2^{m^2+2}< \delta/2$.  Therefore, we recover the payments computed on the original quality values by rounding to the closest multiple of $\delta$. 

\section{Proofs Section \ref{s:egsp}}

\subsection{Proof Claim \ref{claim:low_tight}}

To prove this claim, we first rename the bidders from $1$ to $m$ following the order of the items matched to the bidders. Bidder $i$ is therefore matched to item $i$. Denote by $X_{\geq i}$ the restriction of a set/vector $I$, $J$, $\bf p$ and $\bf q$ to elements $[i+1, \ldots, m]$.  For the proof of the claim we crucially use the following fact. 

\begin{observation}
\label{fact:tight_path}
The alternating path $P_i \subseteq E^=$  that connects each bidder $i$ to the item $m$ with price $p_m=0$  only traverses items from $i$ to $m$.  
\end{observation} 
\begin{proof}
We observe that vectors $\bf p_{\geq i}$, $\bf q_{\geq i}$ form an optimal dual solution for the the problem restricted to bidders $I_{\geq i}$ and items  $J_{\geq i}$ given that the optimal primal solution is formed by a subset of the MWPM computed on $I$ and $J$.  The optimality of the solution when restricted to bidders and items $\geq i$ yields the following fact: 
\end{proof}
We can now prove the claim.
\begin{proof}[Proof of Claim~\ref{claim:low_tight}]
For the base of the induction, the claim is clearly true for bidder $m$.  Assume it is true for all bidders $I_{\geq i+1}$, we prove it holds for bidder $i$.  Observe that, by the inductive hypothesis, bidder $k$ is not tight with any item $>k$. 
Path $P_k$, $k>i$, is formed by the edge $(k,k)$ and then by a path from item $k$ to item $m$ that traverses only bidders $I_{>k}$ and items $J_{>k}$. 
For the inductive hypothesis, by contradiction, assume bidder $i$  is matched with a slot $i$ such that there exists a lower slot $k\geq i$ such that $(i,k)\in E^=$. That must mean  that there exists two distinct alternating paths $P_i \subseteq E^=$ and $P'_i\subseteq E^=$ from bidder $i$ to item $m$.  $P_i$ is the path whose existence is guaranteed by the execution of the Hungarian algorithm. $P_i'$ is formed by the the edge $(i,k)$ followed by the path that connects item $k$ to item $m$. Paths $P_i$ and $P'_{i}$ must form a cycle given that $P_i$ only traverses items $\geq i$ and $P_i'$ only traverses items of $\geq k$, i.e.,  path $P_i$ is not a subset of path $P_k$. 
Given a cycle of tight edges,  we can include in the  MWPM either the odd or the even edges of the cycles while the dual variables $\bf p$ and $\bf q$ are unchanged.  This contradicts the uniqueness of the MWPM and therefore we have a contradiction.
\end{proof}

\subsection{Proof Claim \ref{claim:price_increase}}

\begin{proof}
Given that the quality values $\alpha_{i,j}$ are non increasing, we can move from bids  $b_{i,1}, \ldots, b_{i,m}$ to bids $\tilde b_{i,1}, \ldots, \tilde b_{i,m}$ by considering a sequence of small $\epsilon$ increases of the bids.  We divide the process in two phases. We first bring the values $b_{i,1}, \ldots, b_{i,\tilde j}$ to the final values $\tilde b_{i,1}, \ldots, \tilde b_{i,\tilde j}$, and later, in a second phase, we bring the values  $b_{i,\tilde j +1}, \ldots, b_{i,m}$ to the final values $\tilde b_{i,\tilde j +1}, \ldots, \tilde b_{i,m}$.

For the first phase,  we increase by $\epsilon$  the first  $k$  values  $b_{i,1}, \ldots, b_{i,k}$ for increasing values of $k\leq \tilde j$.  
We prove the claim by showing that the set $E^=_{\tilde j}$ of tight edges does not contract after an $\epsilon$ increase 
on the value of bidder $i$ on the first $k$ slots.  If bidder $i$ is not tight  with any of the  slots in $\{1,\ldots, k\}$, the proof follows immediately from the solution provided by the Hungarian method since the increase of the bids by $\epsilon$ does not violate any constraint $q_i + p_l \geq b_{i,l}$, $l=1, \ldots, k$, if $\epsilon$ is small enough.  If bidder $i$ is tight with at least one of the first $k$ slots, 
we  increase by $\epsilon$ the price of all slots $[1,\ldots, k]$ while the utility $q_i$ is not modified. For any slot $l\in [1,k]$ tight with $i$, it still holds $q_i + p_l = b_{i,l}$.
For all other bidders, we  reduce by $\epsilon$ the utility $q_{i'}$ of  each  bidder $i'$ that is not tight with any slot in $[k+1,\ldots, m]$.  
We observe that all the tight edges for these bidders are still tight. We do not reduce the utility of all bidders, $i$ included,  
that are tight at least with a slot in $[k+1,\ldots, m]$ and therefore the set of tight edges in $E^=_{\tilde j}$ is not contracted. We have also proved that the value of $p_{\tilde j}$ cannot decrease  for each $\epsilon$ increase in the values of the bids. 

Let us now consider the second phase in which we increase by $\epsilon$  for   values  $b_{i,\tilde j+1}, \ldots, b_{i,\tilde j+k}$ for increasing values of $k$ till we reach the final values of the bids.  We know by Claim \ref{claim:low_tight} that bidder $i$ is not tight with any slot $k>\tilde j$ for the final values $\tilde b_{i,1}, \ldots, \tilde b_{i,m}$.   The increase by $\epsilon$ of the   $k$  values  $b_{i,\tilde j+1}, \ldots, b_{i,\tilde j+k}$ will therefore  not make any additional edge tight, and the set $E^=_{\tilde j}$ and the VCG clearing price $p_j$ will stay  unchanged.

%
\end{proof}

\subsection{Proof Claim \ref{claim:price_decrease}}

\begin{proof}
The proof is symmetric to the one of Claim \ref{claim:price_increase}. 
We first bring the values  $b_{i,\tilde j +1}, \ldots, b_{i,m}$ to the final values $\tilde b_{i,\tilde j +1}, \ldots, \tilde b_{i,m}$.   By  Claim \ref{claim:low_tight}, bidder $i$ is not tight with any of the slots in $[\tilde j+1, \ldots, m]$ and therefore the decrease of the values of $\tilde b_{i,\tilde j +1}, \ldots, \tilde b_{i,m}$ will not affect the values of $E^=_{\tilde j}$ and $p_{\tilde j}$. 
In the second phase, we bring the values $b_{i,1}, \ldots, b_{i,\tilde j}$ to the final values $\tilde b_{i,1}, \ldots, \tilde b_{i,\tilde j}$. This second operation will preserve the set $E_{\tilde j}$ of tight edges for all bidders that are tight with some slots $k>\tilde j$ and will not decrease the GSP clearing price $p_{\tilde j}$. 
\end{proof}

\section{Proofs Section \ref{s:swl}}

\subsection{Proof Claim \ref{cl:1}}

\begin{proof}
	We consider two cases for the proof:
	
	\begin{enumerate}
		\item  In a first case, by switching his bid from $b_i$ to $v_i/2$, player $i$ wins
		some slot $k\leq X(\bm{v},i)$. In this case $u_i(v_i/2, \bm{b}_{-i})  \geq 1/2 \alpha_{i,X(\bm{v},i)} v_i$.
		\item  Otherwise, by switching his bid from $b_i$ to $b_i'=v_i/2$, player $i$ wins
		some slot $k\leq X(\bm{v},i)$.  Given that the mechanism computes a MWPM, the player who wins slot $j=X(\bm{v},i)$ under bidding
		profile $(v_i/2,\bm{b}_{-i})$ is such that assigning slot $j$ to bidder $\pi((v_i/2,\bm{b}_{-i}),j)$ and slot $X((v_i/2, \bm{b}_{-i}),i)$ to bidder $i$ 
		gives a reward that is strictly larger than assigning  slot $X(\bm{v},i)$ to bidder $i$ with bid  $b_i'=v_i/2$. 
		In this second case we have 
		
		\begin{eqnarray*} 
		u_i(v_i/2, \bm{b}_{-i}) + \alpha_{\pi((v_i/2,\bm{b}_{-i}),j),j} v_{\pi((v_i/2,\bm{b}_{-i}),j)} &\geq& v_i/2 \times \alpha_{i, X((v_i/2, \bm{b}_{-i}),i)} \\ 
		&&+ \alpha_{\pi((v_i/2,\bm{b}_{-i}),j),j} v_{\pi((v_i/2,\bm{b}_{-i}),j} \\
		&\geq& 1/2 \alpha_{i,X(\bm{v},i)} v_i.\
		\end{eqnarray*}
		
	\end{enumerate}
	
\end{proof}

\subsection{Proof Claim \ref{claim:regretbound}}

\begin{proof}
	By definition, 
	
	\begin{eqnarray}
	R_i(v_i, \bm{b}_{-i}) &=& \max_{b'_i \in \RR_{\geq 0}} \{r_i(b'_i,\bm{b}_{-i}, v_i)  \} \nonumber\\
	&=&\max_{b'_i \in \RR_{\geq 0}} \max \{0, u_i(\mathcal{A}(b'_i,\bm{b}_{-i}),\mathcal{P}(b'_i,\bm{b}_{-i}) ) \nonumber\\
	&\qquad& - u_i(\mathcal{A}(v_i,\bm{b}_{-i}),\mathcal{P}(v_i,\bm{b}_{-i}) ) \} \nonumber\\
	&=& \max_{b'_i \in \RR_{\geq 0}} \max \{0, u_i(\mathcal{A}(b'_i,\bm{b}_{-i}),\mathcal{P}(b'_i,\bm{b}_{-i}) ) \} \nonumber\\
	&\qquad& -  u_i(\mathcal{A}(v_i,\bm{b}_{-i}),\mathcal{P}(v_i,\bm{b}_{-i}) ) \nonumber\\
	&\geq&  \max_{b'_i \in \RR_{\geq 0}} \max \{0, u_i(\mathcal{A}(b_i',\bm{b}_{-i}),\mathcal{P}(b_i',\bm{b}_{-i}) ) \}\nonumber\\
	&\qquad& - u_i(\mathcal{A}(b_i,\bm{b}_{-i}),\mathcal{P}(b_i,\bm{b}_{-i}) ) \nonumber\\
	&=& R_i(b_i, \bm{b}_{-i})
	\end{eqnarray}
	
	where the last inequality follows from 
	$u_i(\mathcal{A}(b_i,\bm{b}_{-i}),\mathcal{P}(b_i,\bm{b}_{-i}) ) \geq u_i(\mathcal{A}(v_i,\bm{b}_{-i}),\mathcal{P}(v_i,\bm{b}_{-i}) )$ that stems for the fact that bid $b_i$ has a regret smaller that truthful strategy $v_i$. 
	
\end{proof}

\subsection{Proof Theorem \ref{thm:swl}}

\begin{proof}
	The proof follows since by Theorem \ref{thm:iff} ~  and ~ Claim \ref{claim:regretbound}, we obtain $E_i(v_i, \bm{b}_{-i}) \geq R_i(v_i, \bm{b}_{-i}) \geq R_i(b_i, \bm{b}_{-i})$.  Therefore: 
	
	\begin{align*}
	\sum_i E_i(v_i, \bm{b}_{-i}) &\geq \sum_i R_i(b_i, \bm{b}_{-i})  \\
	 &\geq \frac12 SW^{OPT} (\bm{v}) - 2 SW(\bm{b})\\
	 &\geq  \frac14 SW^{OPT} (\bm{v}) \\
	&\geq  \frac14 SWL(\bm{b}),
	\end{align*}
	
with the second to last inequality following from $SW^{OPT}(\bm{v}) \geq 8  SW(\bm{b}).$
\end{proof}